\documentclass[manuscript,screen,sigplan,nonacm]{acmart}
\usepackage{graphicx}
\usepackage{algorithmicx}
\usepackage{algorithm} 
\usepackage{algpseudocode}

\usepackage{cleveref}
\usepackage{sidecap}
\usepackage{amsmath}
\DeclareMathOperator*{\argmax}{arg\,max}
\usepackage{multirow}
\usepackage{subcaption}

\usepackage{booktabs}
\usepackage{graphicx}
\usepackage{enumitem}

\AtBeginDocument{%
  \providecommand\BibTeX{{%
    \normalfont B\kern-0.5em{\scshape i\kern-0.25em b}\kern-0.8em\TeX}}}

\begin{document}

\acmYear{2024}\copyrightyear{2024}
\setcopyright{acmlicensed}
\acmConference[MIDDLEWARE '24]{24th International Middleware Conference}{December 2--6, 2024}{Hong Kong, Hong Kong}
\acmBooktitle{24th International Middleware Conference (MIDDLEWARE '24), December 2--6, 2024, Hong Kong, Hong Kong}
\acmDOI{10.1145/3652892.3700788}
\acmISBN{979-8-4007-0623-3/24/12}

\newcommand{\ICC}{Internet Computer Consensus}
\newcommand{\iccin}{icc}
\newcommand{\ICCal}{\textsf{ICC}}
\newcommand{\iccal}{\textsf{icc}}

\newcommand{\FICC}{\textsc{Banyan}}
\newcommand{\ficcin}{banyan}
\newcommand{\FICCal}{\textsc{Banyan}}
\newcommand{\ficcal}{\textsc{banyan}}

\newcommand{\rp}[1]{{\left(#1\right)}}
\newcommand{\qp}[1]{{\left[#1\right]}}
\newcommand{\cp}[1]{{\left\{#1\right\}}}
\newcommand{\ap}[1]{{\left\langle#1\right\rangle}}

\newcommand{\abs}[1]{{\left|#1\right|}}
\newcommand{\floor}[1]{{\lfloor #1 \rfloor}}
\newcommand{\ceil}[1]{{\lceil #1 \rceil}}

\newcommand{\bin}[3]{
    \ifthenelse
    {\equal{#3}{}}
    {\text{Bin}\qp{#1, #2}}
    {\text{Bin}\qp{#1, #2}\rp{#3}}
}

\newcommand{\pois}[2]{
    \ifthenelse
    {\equal{#2}{}}
    {\text{Poisson}\qp{#1}}
    {\text{Poisson}\qp{#1}\rp{#2}}
}

\newcommand{\bern}[2]{
    \ifthenelse
    {\equal{#2}{}}
    {\text{Bern}\qp{#1}}
    {\text{Bern}\qp{#1}\rp{#2}}
}

\newcommand{\prob}[1]{{\mathcal{P}\qp{#1}}}

\newcommand{\true}{\text{\tt True}}
\newcommand{\false}{\text{\tt False}}

\newcommand{\powerset}[2]{
    \ifthenelse
    {\equal{#2}{}}
    {\mathbb{P}\rp{#1}}
    {\mathbb{P}^{#2}\rp{#1}}
}

\newcommand\et{\;\textbf{and}\;}
\newcommand\ou{\;\textbf{or}\;}

\NewDocumentCommand{\Break}{}{\State \textbf{break}}

\NewDocumentCommand{\event}{mmo}{
    \IfValueTF{#3}
    {{\ap{#1.\textrm{#2} \mid #3}}}
    {{\ap{#1.\textrm{#2}}}}
}

\NewDocumentCommand\newinstance{m}{{\;\textbf{new}\; \text{#1}}}

\algnewcommand\Instance[2]{\State #1, \textbf{instance} #2}
\algnewcommand\MultipleInstances[1]{\State #1, \textbf{multiple instances}}
\algnewcommand\Exposes[3]{\State #1 \textbf{exposes} #2 \textbf{as} #3}

\algnewcommand\Alias[2]{\State \textbf{use} $#1$ \textbf{as alias for} $#2$}

\algnewcommand\Trigger[3]{\State \textbf{trigger} $\event{#1}{#2}[#3]$}
\algnewcommand\TriggerPure[2]{\State \textbf{trigger} $\event{#1}{#2}$}
\algnewcommand\SetTimeout[2]{\State \textbf{set timeout} $#1$ \textbf{in} $#2$}

\algsetblockdefx[Implements]{Implements}{EndImplements}{}{}{\textbf{Implements:}}{}
\algsetblockdefx[Uses]{Uses}{EndUses}{}{}{\textbf{Uses:}}{}
\algsetblockdefx[Parameters]{Parameters}{EndParameters}{}{}{\textbf{Parameters:}}{}
\algsetblockdefx[Optimizations]{Optimizations}{EndOptimizations}{}{}{\textbf{Optimizations:}}{}

\algsetblockdefx[Upon]{Upon}{EndUpon}{}{}[3]{\textbf{upon event} $\event{#1}{#2}[#3]$ \textbf{do}}{}
\algsetblockdefx[UponPure]{UponPure}{EndUponPure}{}{}[2]{\textbf{upon event} $\event{#1}{#2}$ \textbf{do}}{}
\algsetblockdefx[UponTimeout]{UponTimeout}{EndUponTimeout}{}{}[1]{\textbf{upon timeout} $#1$ \textbf{do}}{}
\algsetblockdefx[UponExists]{UponExists}{EndUponExists}{}{}[2]{\textbf{upon exists} $#1$ \textbf{such that} $#2$ \textbf{do}}{}
\algsetblockdefx[UponCondition]{UponCondition}{EndUponCondition}{}{}[1]{\textbf{upon} $#1$ \textbf{do}}

\algsetblockdefx[Procedure]{Procedure}{EndProcedure}{}{}[2]{\textbf{procedure} \emph{#1}($#2$) \textbf{is}}{}

\algsetblockdefx[IfExists]{IfExists}{EndIfExists}{}{}[2]{\textbf{if exists} $#1$ \textbf{such that} $#2$ \textbf{then}}{\textbf{end if}}
\algsetblockdefx[While]{While}{EndWhile}{}{}[1]{\textbf{while} $#1$ \textbf{do}}{\textbf{end while}}
\algdef{SE}[DoUntil]{DoUntil}{EndDoUntil}{\textbf{do}}[1]{\textbf{until} $#1$}
\algsetblockdefx[For]{For}{EndFor}{}{}[3]{\textbf{for} $#1 \;\textbf{in}\; #2..#3$ \textbf{do}}{\textbf{end for}}
\algsetblockdefx[ForAll]{ForAll}{EndForAll}{}{}[2]{\textbf{for all} $#1 \;\textbf{in}\; #2$ \textbf{do}}{\textbf{end for}}
\algsetblockdefx[ForAllSuchThat]{ForAllSuchThat}{EndForAllSuchThat}{}{}[3]{\textbf{for all} $#1 \;\textbf{in}\; #2$ \textbf{such that} $#3$ \textbf{do}}{\textbf{end for}}
\algsetblockdefx[NTimes]{NTimes}{EndNTimes}{}{}[1]{\textbf{for} $#1$ \textbf{times do}}{\textbf{end for}}

\title{\FICC: Fast Rotating Leader BFT}

\author{Yann Vonlanthen}
\email{yvonlanthen@ethz.ch}
\affiliation{%
  \institution{ETH Zurich}
  \country{Switzerland}
}
\author{Jakub Sliwinski}
\email{jsliwinski@ethz.ch}
\affiliation{%
  \institution{ETH Zurich}
  \country{Switzerland}
}
\author{Massimo Albarello}
\email{massimo@onfabric.io}
\affiliation{%
  \institution{ETH Zurich}
  \country{Switzerland}
}

\author{Roger Wattenhofer}
\email{wattenhofer@ethz.ch}
\affiliation{%
  \institution{ETH Zurich}
  \country{Switzerland}
}

\renewcommand{\shortauthors}{Y. Vonlanthen, J. Sliwinski, M. Albarello, R. Wattenhofer}

\begin{abstract}
This paper presents \FICC, the first rotating leader state machine replication (SMR) protocol that allows transactions to be confirmed in just a single round-trip time in the Byzantine fault tolerance (BFT) setting. Based on minimal alterations to the Internet Computer Consensus (ICC) protocol and with negligible communication overhead, we introduce a novel dual mode mechanism that enables optimal block finalization latency in the fast path. Crucially, the modes of operation are integrated, such that even if the fast path is not effective, no penalties are incurred. Moreover, our algorithm maintains the core attributes of the ICC protocol it is based on, including optimistic responsiveness and rotating leaders without the necessity for a view-change protocol.

We prove the correctness of our protocol and provide an open-source implementation of it. \FICC\ is compared to its predecessor ICC, as well as other well known BFT protocols, in a globally distributed wide-area network. Our evaluation reveals that \FICC\ reduces latency by up to 30\% compared to state-of-the-art protocols, without requiring additional security assumptions. 

\end{abstract}



\maketitle

\section{Introduction}
\label{sec:introduction}
Byzantine fault tolerance (BFT) is a class of protocols that provide guarantees in the presence of arbitrary faults, such as a powerful adversary controlling both a share of participants and the network scheduling~\cite{lamport1982byzantine}. BFT protocols can be applied whenever a fixed (also called permissioned) set of untrusted parties desires to reach an agreement.

This core primitive has enabled the creation of decentralized protocols that are akin to a world computer, enabling anyone to run Turing-complete programs. Various such world computers exist today~\cite{buterin2013ethereum, dfinityIC, filecoin, blackshear2023sui}, with varying trade-offs regarding performance, security, and inclusiveness. At each system's core, lays an ordering protocol assuring that resource accesses are consistent across all participants (called replicas), thus guaranteeing deterministic and secure execution of the computation.

Byzantine agreement (BA) protocols typically provide consensus on a single decision, while state machine replication protocols (SMR) focus on making efficient use of this costly primitive. 
In practice, the most widely used consensus protocols (such as Bitcoin~\cite{nakamoto2009bitcoin}, Ethereum~\cite{buterin2013ethereum}, and Algorand~\cite{gilad2017algorand}) make use of an elected leader to propose a batch of transactions, called a block. A total order across leaders is then obtained by chaining blocks.

There are a few considerations that are especially crucial to such leader-based protocols. Firstly, leaders can misbehave and propose conflicting blocks. While this can be detected, the resulting change of the leader (called \textit{view-change} in the literature) leads to notoriously high complexity. Honest, but slow leaders are just as problematic, as they cannot easily be called out, but can stall the effective throughput nonetheless~\cite{allen2009leaderslow}.
Moreover, relying on a single leader leads to uneven load, both in terms of computation and communication~\cite{crain2018dbft}, and might facilitate censorship. Finally, leaders sometimes extract value from their control of the block creation. On the Ethereum blockchain, for example, miner-extractable value (MEV) is a major source of income for replicas and leads to potentially dangerous incentives to deviate from honest behavior~\cite{daian2020flash}. These issues can at least partially be resolved by employing rotating or random leaders~\cite{camenisch2022}, which is the class of protocols we aim to improve upon in this work.

Specifically, we address what we believe to be a major roadblock in the adoption of decentralized computers: high latency. In contrast to privacy and security, latency is at the forefront of the user experience. While centralized and trusted applications (such as credit card payments or social media) provide almost instantaneous services, their Byzantine fault-tolerant counterpart typically comes at the cost of a noticeable delay, in some cases in the order of minutes or even hours (such as in Bitcoin). Furthermore, contrary to other metrics such as throughput, latency cannot be readily improved by hardware upgrades.

\begin{figure}[]
\centering
\includegraphics[width=0.45\textwidth]{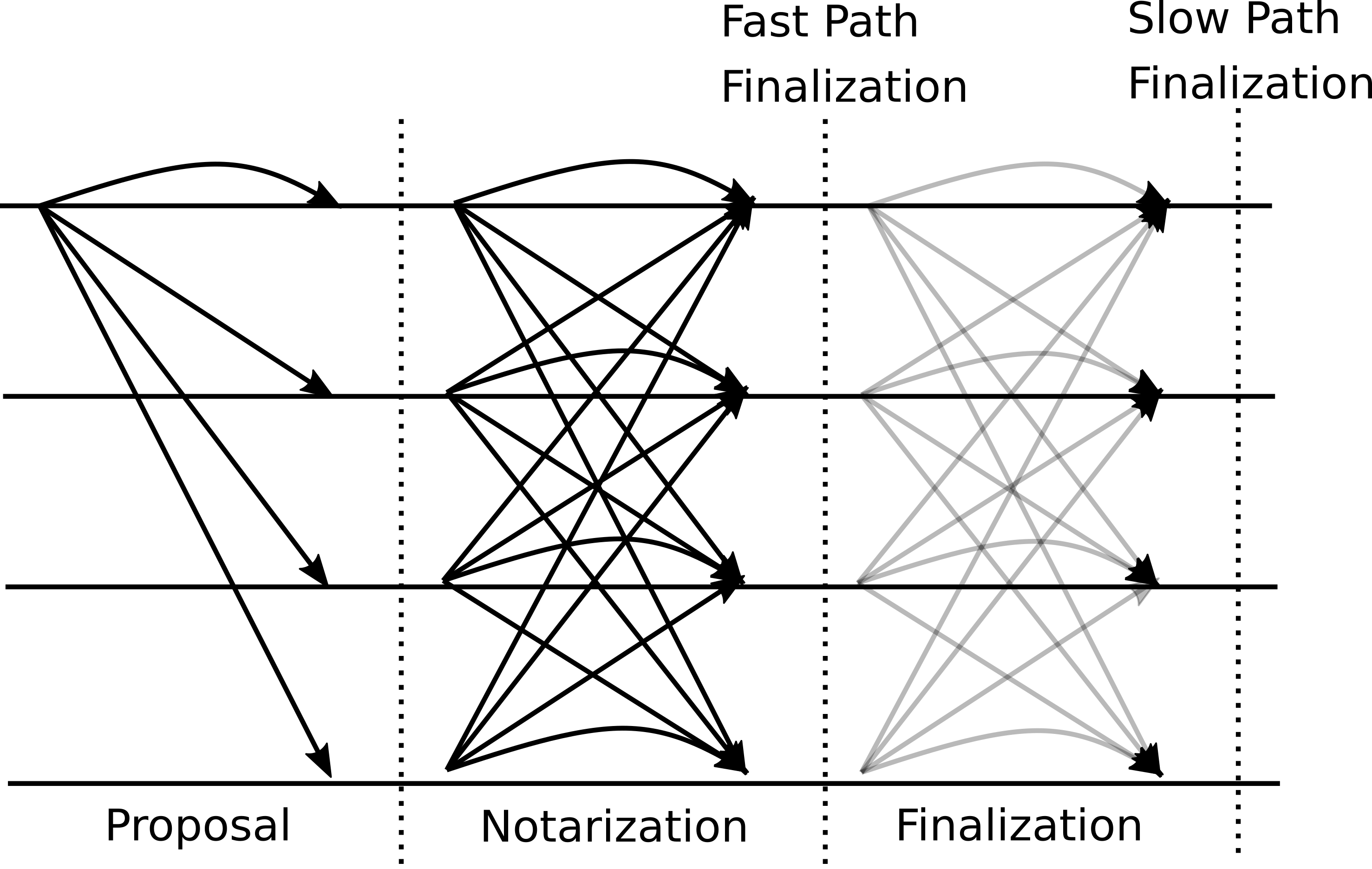}
\caption{\FICC\ can terminate after just two communication steps. Existing rotating leader BFT protocols require at least three communication steps.} 
\label{fig:communication_rounds}
\end{figure}

In this work, we present the first rotating leader SMR protocol that allows transactions to be confirmed in just a single round-trip time. This is optimal, as to achieve fault tolerance the block must 1) reach a fraction of replicas, that have to 2) respond to acquiesce to the processing~\cite{song2008bosco}. This is visualized in \Cref{fig:communication_rounds}.

We are able to achieve this by uniting two lines of research. The first deals with the theoretical bounds of latency, and has established optimally fast single-shot protocols in the good case~\cite{kursawe2002optimistic, martin2006fast, aublin2015next, kuznetsov2021revisiting, abraham2021good}. The second, more well-known line of research, provides a class of consensus protocols that solve state machine replication, in an efficient and fair way, by employing rotating or random leaders. Prominent examples include HotStuff, Tendermint, Casper, Algorand and Bullshark~\cite{yin2019hotstuff, buchman2016tendermint, buterin2017casper, gilad2017algorand, spiegelman2022bullshark}.

As Chan and Pass point out~\cite{chan2023simplex}, rotating leader protocols can be differentiated by either favoring \textit{proposal confirmation time}, meaning block finalization latency, or \textit{block creation time}, corresponding to the delay between the blocks and thus the associated throughput. The Internet Computer Consensus family of protocols~\cite{camenisch2022} excels in both of these metrics.

\textbf{Our contribution.}
We present the \FICC\ protocol that achieves optimal consensus latency while preserving the state-of-the-art characteristics of ICC.
\begin{enumerate}
    \item \FICC\ is the first SMR protocol with rotating leaders supporting two-step chain growth and two-step finalization under partial synchrony. We prove the safety and liveness of \FICC\ for $n \geq 3f + 2p^* - 1$, where $n$ is the total number of replicas in the network, and $f$ is the maximum number of Byzantine replicas tolerated. Additionally, $p^* \ge 1$ is a parameter that can be set freely ($p^* \leq f$), and determines the effectiveness of the fast path. We prove that as long as no more than $p^*$ replicas are unresponsive, \FICC\ finalization latency is just two network delays.
    \item We provide a proof-of-concept \FICC\ implementation written in Golang, built on top of the Bamboo BFT framework~\cite{gai2021bamboo}. %
    \item We measure the proposal finalization time in three different wide area networks. We empirically demonstrate that, compared to ICC, HotStuff, and Streamlet, \FICC\ achieves the fastest proposal finalization time, improving by up to 30\% over the runner-up. Further, we show that \FICC\ does not suffer from increased variance in its latency, and behaves as ICC under crash-faults.
\end{enumerate}

\section{Related Work}

Byzantine fault tolerance was introduced by Lamport et al. in 1982~\cite{lamport1982byzantine}. A long line of work studies consensus protocols in the permissioned setting under various synchrony assumptions~\cite{ben1983another, bracha1987asynchronous, lamport2003lower}). The partially synchronous network model was first introduced in~\cite{dwork1988consensus}, together with the first consensus protocol for this setting. The first truly practical protocol in the partially asynchronous setting is the PBFT protocol, as described in~\cite{castro1999practical}.

\textbf{Fast Consensus.} Our algorithm relies on a long line of established work on reaching consensus in two communication steps (or one round trip time), also referred to as fast or early-stopping consensus~\cite{brasileiro2001consensus, kursawe2002optimistic, friedman2005simple, song2008bosco, martin2006fast, guerraoui2007refined}. As shown by Brasileiro et al.~\cite{brasileiro2001consensus} in the crash-fault model, it is possible to terminate early if ``enough'' acknowledgments are received. A year later, a fast path was presented by Kursawe~\cite{kursawe2002optimistic}, in which a two-step fast path is paired with a subprotocol in the slow path. Upon the expiration of a timer, a fallback subprotocol is used to ensure liveness. 

Song et al. describe Bosco~\cite{song2008bosco}, an algorithm that provides a two-step fast path to any underlying consensus algorithm. Their fast path does not rely on synchrony assumptions, and, assuming $n > 5f$, is triggered when all servers are in pre-agreement, i.e., propose the same value. Assuming $n > 7f$, the fast path is triggered when all honest servers are in pre-agreement instead. Using similar building blocks it was later shown that consensus can be performed ``on demand'', i.e., the fallback is only used when the fast path does not succeed~\cite{sliwinski2022}. %
As \Cref{fig:comparison} shows, these approaches come at the cost of a longer recovery during the fallback on the slow path though, as the slow path can only be started once the fast path round is over. In this work, we show that both the (single) fast path round and the first slow path round can happen simultaneously (with only constant message overhead). 

\begin{figure}
\includegraphics[width=0.9\columnwidth]{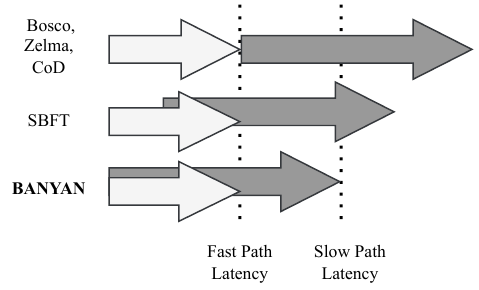}
\centering
\caption{The different approaches of fast path protocols. In many protocols, the slow path starts after the fast path fails. In \FICC, the fast path is integrated into the slow path.} 
\label{fig:comparison}
\end{figure}

Earlier, Martin and Alvisi introduce Fast Byzantine Paxos (FaB Paxos)~\cite{martin2006fast}, a family of protocols parameterized not only by $f$ but also $p$ ($0 \leq p \leq f$), the number of (non-leader) replicas that are not needed for the fast path to succeed. Their algorithms provide Byzantine agreement, when $n \geq 3f + 2p + 1$, where the fast path triggers as long as $n-p$ servers behave honestly and are in pre-agreement.
Instead of relying on a variant of PBFT in the slow path, we integrate our protocol with ICC, sidestepping the complexity introduced by view changes. 

The lower bound by Martin and Alvisi~\cite{martin2006fast} has recently been improved to $n \geq max(3f+2p -1, 3f+1)$ by Kuznetsov et al.~\cite{kuznetsov2021revisiting} and Abraham et al.~\cite{abraham2021good}. Their insight is that misbehaving (i.e., equivocating) leaders can be reliably detected, and thus acceptors can wait for $n-f$ votes, \textit{excluding} the malicious leader ($n-f+1$ votes in total). To the best of our knowledge, we present the first protocol that matches this lower bound without the need for a \textit{view-change} protocol. %

\textbf{Fast SMR.}
Ideas from the FaB protocol were borrowed by the Zyzzyva protocol~\cite{kotla2007zyzzyva} as well, which turns the single-shot FaB consensus into an SMR protocol. Improving upon Zyzzyva, Aliph~\cite{aublin2015next} reduces the latency to two steps. These protocols do not allow running both the fast path and a slow path concurrently, though, and thus suffer from switching costs (see \Cref{fig:comparison}). Zyzzyva safety issues were pointed out by~\cite{abraham2017revisiting} and led to the development of Zelma~\cite{abraham2018zelma}, combining FaB and Zyzzyva's benefits. 

Zelma is also at the core of the SBFT protocol~\cite{gueta2019sbft}. SBFT is the first protocol running both pessimistic and optimistic paths simultaneously in a dual mode. However, the fast path of SBFT has one more communication step than \FICC, and assumes $n \ge 3f + 2p + 1$. Moreover, the slow path of SBFT is not optimistically responsive, as it triggers only after a time-out (\cref{fig:comparison}).
A large body of work has based itself on the 3-phase (6-step) HotStuff~\cite{yin2019hotstuff} protocol and presented improvements upon it~\cite{jalalzai2020fast,malkhi2023hotstuff, kang2024hotstuff}. %
Jolteon and Ditto~\cite{gelashvili2022jolteon} have improved the latency while adding an asynchronous fallback to enhance its performance during epoch synchronization (see \Cref{tab:related_work}).  

Most recently Malkhi and Nayak propose a 2-phase (4-step) protocol and a simpler fast path~\cite{malkhi2023hotstuff}. These protocols reach no lower than $3 \delta$ proposer latency in the fast path. %

\begin{table*}[ht]
\centering
\begin{tabular}{@{}lcccccc@{}}
\toprule
& \begin{tabular}[c]{@{}c@{}}Block \\ Finalization \\ Latency\end{tabular} & \begin{tabular}[c]{@{}c@{}}Block \\ Finalization\\ Requirement\end{tabular}  & \begin{tabular}[c]{@{}c@{}}Block \\ Creation \\ Latency\end{tabular} & \begin{tabular}[c]{@{}c@{}}Block \\ Creation \\ Requirement\end{tabular}                     & \begin{tabular}[c]{@{}c@{}}Number of \\ Replicas \textit{equals}\end{tabular} & \begin{tabular}[c]{@{}c@{}}Supports \\ Rotating \\ Leaders\end{tabular} \\ \midrule
Casper FFG~\cite{buterin2017casper}                                                                  & $O(\Delta)$                                                 & $2f + 1$                                           & $O(\Delta)$                                                     & $\text{N/A}^\dagger$                                                                                     & \begin{tabular}[c]{@{}l@{}}$3f+1$ \end{tabular}     & \checkmark                                                              \\ 
\begin{tabular}[c]{@{}l@{}} Fast HotStuff~\cite{jalalzai2020fast}$^\ddagger$ \end{tabular} & $5 \delta$                                                  & $2f+1$                                                              & $2 \delta$                                                      & $2f+1$                                                                                    & $3f+1$                                                       &                                                                         \\  
\begin{tabular}[c]{@{}l@{}}Jolteon~\cite{gelashvili2022jolteon} \end{tabular} & $5 \delta$                                                  & $2f+1$                                                              & $2 \delta$                                                      & $2f+1$                                                                                    & $3f+1$                                                       &                                                                         \\  

PaLa~\cite{chan2018pala}                                                                      & $4 \delta$                                                  & $2f + 1 $                                           & $2 \delta$                                                      & $2f + 1 $                                                                 & $3f+1$                                                       &                                                                         \\  
Zelma~\cite{abraham2018zelma}                                                                       & $2 \delta$                                                  & $3f + p + 1$                                                             & $2 \delta$                                                      & $2f + p + 1$                                                                               & $3f + 2p + 1$                                                &                                                                         \\  
SBFT~\cite{gueta2019sbft}                                                                       & $3 \delta$                                                  & $3f + p + 1$                                                        & $3 \delta^\mathsection$                                                    & $2f + p + 1$                                                                              & $3f + 2p + 1$                                                &                                                                        \\

\begin{tabular}[c]{@{}l@{}}Streamlet~\cite{chan2020streamlet} \end{tabular}                                                                       & $6 \Delta$                                                  & $2f + 1$                                                               & $2 \Delta$                                                      & $2f + 1$                                                                                     & $3f+1$                                                       & \checkmark        
\\

Bullshark~\cite{spiegelman2022bullshark}                                                                   & $4 \delta ^\mathparagraph$                                                 & $2f+1$ & $2 \delta$                                                     & $2f+1$  & $3f+1$                                                       & \checkmark                                                              \\  
BBCA-Chain~\cite{malkhi2023bbca}                                                                   & $3 \delta$                                                 & $2f+1$ & $3 \delta$                                                     & $2f+1$ & $3f+1$                                                       & \checkmark                                                              \\  

\begin{tabular}[c]{@{}l@{}}ICC~\cite{camenisch2022} / Simplex~\cite{chan2023simplex} \end{tabular}                                                                       & $3 \delta$                                                  & $2f + 1$                                                               & $2 \delta$                                                      & $2f + 1 $                                                                                     & $3f+1$                                                       & \checkmark                                                            

                                         \\  
\begin{tabular}[c]{@{}l@{}}Mysticeti~\cite{babel2023mysticeti} \end{tabular}                                                                       & $3 \delta$                                                  & $2f + 1$                                                               & $1 \delta$                                                      & $2f + 1$                                                                                     & $3f+1$                                                       & \checkmark        
\\
\midrule
\FICCal                                                                      & ${2 \delta}$                                                  & $3f + p^* - 1$                                                             & $2 \delta$                                                      & $2f + p^*$          & $3f + 2p^* - 1$                                                & \checkmark                                                              \\ \bottomrule
\end{tabular}%
\caption{Popular and Fast State Machine Replication Protocols. To simplify comparison, we assume that the number of replicas is equal to the respective lower bound. $\Delta$ denotes the message delivery time upper bound, while $\delta$ is the true message delivery time. 
$^\dagger$ Non-equivocation is enforced by slashing. 
$^\ddagger$ We consider the pipelined version of Fast HotStuff. 
$^\mathsection$ To the best of our knowledge no pipelining is specified for SBFT. 
$^\mathparagraph$ We consider Bullshark's best case latency (anchor blocks). 
$^*$ For simplicity, we replace $p$ by $p^*, p^* \geq 1$.}
\label{tab:related_work}
\end{table*}

\textbf{Rotating Leader BFT protocols.}  Many recent advances in BFT protocols rely on rotating leaders~\cite{buchman2016tendermint, buchnik2019fireledger, yin2019hotstuff, chan2020streamlet}, as originally introduced by Veronese et al.~\cite{veronese2009spinning}. Mir-BFT~\cite{stathakopoulou2019mir} and its follow-up work~\cite{stathakopoulou2022state} for example, improve upon sequential-leader approaches, by running PBFT instances on a set of leaders. When a leader is slow, it is replaced. 
Other SMR protocols are designed to perform well under attack~\cite{amir2010prime, avarikioti2023fnf}. However, in the permissioned blockchain setting it is often assumed that during periods of synchrony, misbehavior can be punished, either by exclusion or slashing~\cite{buterin2017casper}. As such, we focus on the latency in the optimistic case and satisfy the same pessimistic liveness as ICC, which was shown to be adequate in~\cite{chan2023simplex}. 
Recent concurrent work by Chan and Pass~\cite{chan2023simplex} provides a theoretical framework for this class of algorithms, and \Cref{tab:related_work} is partially inspired by it. Their proposed protocol called Simplex falls into the same category as the ICC and \FICC\ algorithms but improves \textit{pessimistic liveness} by only allowing a single leader per round. We believe that a technique similar to the one shown in this work could be applied to bring our fast path to Simplex too. 
In~\cite{abraham2021optimal} Abraham et al. prove that $2\Delta$ is the upper and lower bound for the good-case latency of rotating leader BFT protocols, in the synchronous setting. In this work, we provide a matching upper bound in the partially synchronous setting.

\textbf{DAG-based BFT protocols.} Another line of work proposes to improve the throughput of leader-based protocols by allowing all parties to broadcast transactions simultaneously. Many so-called DAG-based protocols have emerged recently~\cite{keidar2021alldag, spiegelman2022bullshark, babel2023mysticeti, malkhi2023bbca}.

One main idea is to disconnect transactions broadcasting from finalization~\cite{danezis2022narwhal}, and another is to allow the finalization of blocks outside the main blockchain by causally referencing them~\cite{lewenberg2015inclusive, keidar2021alldag}. In Bullshark~\cite{spiegelman2022bullshark}, all replicas are parallel leaders and broadcast one batch of transactions in each round using Byzantine Consistent Broadcast. While doing so, they include references to at least $n-f$ blocks from the previous round. A more recent protocol called BBCA-Chain~\cite{malkhi2023bbca} reduces Bullshark's latency by removing the need for specialized block layers and instead broadcasts blocks using a primitive called Byzantine Broadcast with Complete-Adopt (BBCA). Mysticeti~\cite{babel2023mysticeti} forgoes the confirmation of blocks and instead relies solely on the DAG edges as input to the protocol, leading to better optimistic latency. ICC does not allow causally referencing non-leader blocks to finalize them and thus does not benefit from the same throughput advantages. However, we believe ICC is a good fit to describe our fast path mechanism, and hypothesize that it applies to DAG protocols too.

In \Cref{tab:related_work} an overview of state-of-the-art protocols is shown. For both the block finalization latency and the block creation latency, we have included the required number of replicas that have to respond in order to proceed in the optimistic case (synchronous round with a block proposed by an honest leader). While irrelevant in the calculation of the theoretical latency bounds, we found that small changes in these requirements could have large effects in practice, especially on global-scale deployments as they are seen today. (Intuitively, it is not uncommon for a few outlier replicas to have a higher latency. If progress can be made safely without them, a large performance increase can be experienced.)

\textbf{Other aspects.} Many works such as~\cite{malkhi2023hotstuff, camenisch2022, gueta2019sbft} also focus on the message complexity of the fast path and the view change. We do describe how BLS signature aggregation~\cite{boneh2018compact} can be used for \FICC, but leave further advancements to future work. Message complexity and performance do not always go hand in hand in practice~\cite{danezis2022narwhal,camenisch2022} and oftentimes probabilistic broadcast through gossiping is preferred~\cite{libp2p}.

\section{Model}
We consider a network of $n \geq \max{(3f + 2p - 1, 3f + 1)}$ participants called replicas. Our protocol is safe and live with up to $f$ replicas being Byzantine. Additionally, we guarantee \textit{fast termination} when up to $p$ replicas are not cooperating. In other words, the parameter $p \in [0, f]$ denotes the maximum number of replicas that are not needed for the fast path to be successful. Note that there is no reason to choose $p = 0$, as our protocol will be strictly faster with $p=1$, and require the same number of replicas. We can thus also write $n \geq 3f + 2p^* - 1$, where $p^* \in [1, f]$.

By setting $p=1$, we reach the upper bound on the number of Byzantine replicas permissible~\cite{dwork1988consensus}, i.e. $n \geq 3f + 1$. In this case, the fast path will be used when the leader and all but one replica behave honestly. On the other side of the spectrum, given an honest leader, the fast path can be made robust against Byzantine behavior, by setting $p=f$. 

We consider the same communication model used in ICC, in which replicas communicate over a partially synchronous network~\cite{dwork1988consensus}. 
In a synchronous network, there is a known fixed upper bound $\Delta$ on the time required for a message to be sent from one party to another. In an asynchronous network, no fixed upper bound $\Delta$ exists. We use $\delta$ to denote the \textit{true message delivery time} of replicas, i.e., the unknown time it takes for one communication step across all replicas.
In partial synchrony, the network alternates intervals of synchrony, in which the bound $\delta < \Delta$ holds, to ones of asynchrony. 

Replicas communicate via all-to-all authenticated links. We assume the existence of a public-key infrastructure (PKI), secure digital signatures, and collision-resistant hash functions. Additionally, we assume replicas have access to shared randomness, e.g., through a safe and live random beacon protocol~\cite{raikwar2022randombeacon}. 

\section{Internet Computer Consensus}
\label{sec:icc}
The slow path of the \FICC\ algorithm corresponds almost precisely to the \ICC\ protocol (\ICCal), which we thus introduce first.

\begin{figure*}[h]
\centering
\includegraphics[width=0.85\textwidth]{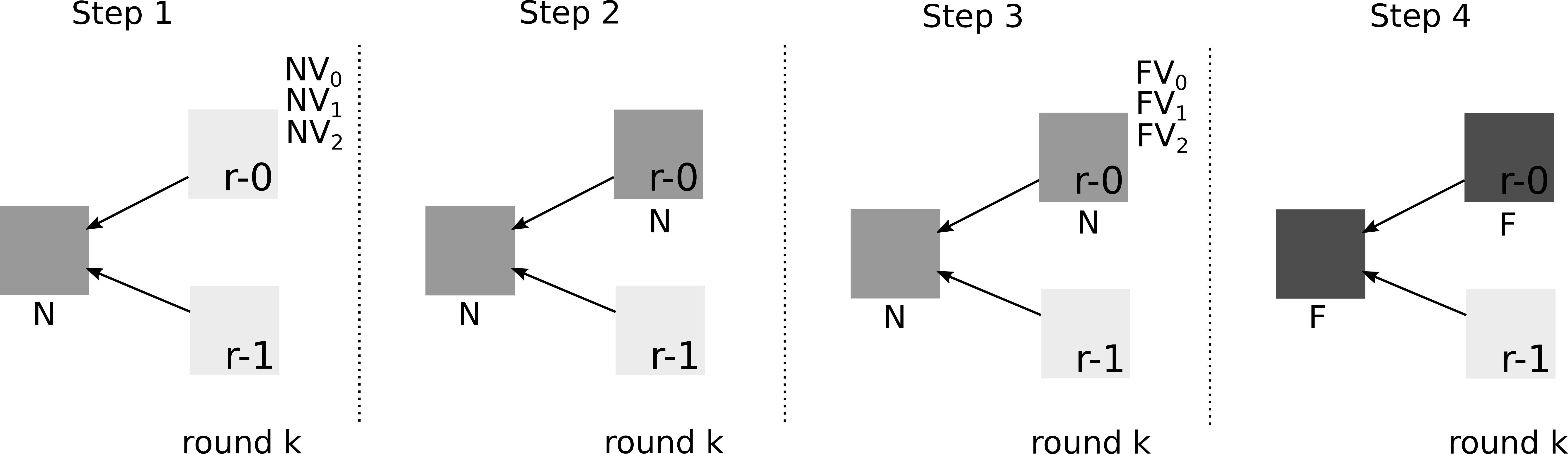}
\caption{The diagram shows the steps necessary to reach the finalization of a block. Black blocks are \textit{finalized}, dark gray blocks are \textit{notarized}, while light gray blocks are not.  In this example, $n=4$ and $f=1$. Initially, $n-f=3$ replicas send notarization votes (NV) for the rank-0 block of round $k$. Replicas will not send a NV for higher rank proposals, in this round. As soon as the NV are received, the replica combines them in a notarization (N). Then, the replicas that only sent an NV for the rank-0 block will also send a finalization vote (FV) for it. Once the three FV are received, they will be combined into a finalization (F) and both the r-0 block and its ancestor(s) are finalized.} \label{icc_finalization}
\end{figure*}

\textbf{Overview.} The goal of the \ICCal\ protocol is to provide a total ordering (i.e., chain) of blocks among replicas in a network. As the protocol advances, a tree of blocks is constructed, starting from a genesis block that is at the root of the tree. Blocks are added to the block-tree if they are safe to be extended.
Each replica may have a different, partial view of the block-tree.
Replicas make progress, by \textit{finalizing} at most one block per tree height. By traversing 
the tree edges, a path of finalized blocks is established. This path constitutes the blockchain. Honest replicas are guaranteed to have a common prefix of this path.

\textbf{Creating a Block-Tree through Notarization.} The ICC protocol proceeds in \textbf{rounds}. In the k-th round, one or more blocks are added to a block-tree at height $k$. %
To be added to the tree, a block must be \textit{notarized}. In ICC, a block $b$ becomes \textbf{notarized} for a replica, once at least $n-f$ \textit{notarization votes} for block $b$ are received. A \textbf{notarization vote} is a BLS signature sent by a replica $u$ for a block $b$, implying that replica $u$ validated block $b$. Notarization votes can be aggregated to a single multi-signature, that can be efficiently verified~\cite{boneh2018compact}. As soon as block $b$ is validated by at least $n-f$ replicas, it becomes notarized, and can be added to the block-tree.

\textbf{Block Proposal.} In principle, in every round $k$, each replica can propose a block and all notarized blocks make it into the block-tree. Therefore, at any height $k$ there can be up to $n$ blocks. However, this makes it difficult for the replicas to agree on which block should be part of the blockchain, as at most one block per height should be included. To reduce the number of blocks proposed in each round, a \textit{random beacon} is used to generate a random permutation of the $n$ replicas. The permutation defines a different \textbf{rank} $r \in [0,n-1]$ for each replica for that particular round. The replica with the lowest rank, i.e., the \textbf{rank-0 replica}, is the \textbf{leader} of that round. At the beginning of a round, each replica computes its rank and starts a timer. The leader of that round proposes the block immediately, while the other replicas wait for a time directly proportional to their rank before proposing a block. Let $r$ be the rank of a replica in a given round, the replica delays the proposal of its block by a proposal delay $\Delta_{prop}(r) = 2\Delta \times r$. In round $k$, when deciding which block in the block-tree to extend, the replica considers only notarized blocks at height $k-1$.

\textbf{Notarization.} In round $k$, before sending a notarization vote for a block $b$ proposed by a replica $u$ with rank $r^{(u)}$, the replica $v$ receiving the block $b$ waits a notarization delay of $\Delta_{notary}^{(u)} = 2\Delta \times r^{(u)}$ after starting round $k$. Therefore, once replica $v$ receives a block $b$ from replica $u$, it checks if it is time to send a notarization vote for a block of rank $r^{(u)}$. If so, it broadcasts the notarization vote to all other replicas, otherwise, it waits until $\Delta_{notary}^{(u)}$ has passed.
This way, blocks of lower rank should become notarized and added to the block-tree before others. Once a block $b$ in round $k$ becomes notarized for replica $u$, 
replica $u$ broadcasts the notarization to other replicas, stops sending notarization votes for blocks of round $k$, and starts round $k+1$. Thus, when the leader is honest, the network is in a phase of synchrony and the block proposal and notarization delays are set accordingly, only the block proposed by the rank-0 replica will be added to the tree at height $k$. If the leader is not honest or the network is asynchronous, some other replicas of higher rank may also propose blocks, and also have their blocks added to the tree. 

\textbf{Finalization.} As there might be multiple notarized blocks at the same height, we must guarantee that all replicas agree on the same \textit{finalized} block and that all \textit{finalized} blocks are part of the same blockchain. 
ICC uses a mechanism similar to notarization to guarantee that no consistency violation across rounds is possible. If the replica did not broadcast notarization votes for any other block than $b$ in the same round $k$, the replica will also broadcast another BLS-signature for block $b$, called \textbf{finalization vote}. If a replica $u$ receives at least $n-f$ finalization votes for the same block $b$, the replica $u$ can aggregate them into a \textbf{finalization}. At this point, the replica $u$ is guaranteed that all replicas will eventually include block $b$ in the blockchain. Once a finalization is created for a given block $b$, block $b$ is said to be \textbf{explicitly finalized} and can be output. The finalization is then sent to all the other replicas.

The rule for broadcasting finalization votes implies that in some rounds, no block might become finalized by the aggregation of at least $n-f$ finalization votes. However, as soon as a block in a later round is explicitly finalized, all its ancestors in the tree, until the last explicitly finalized block, automatically become \textbf{implicitly finalized}.
An illustration of the steps necessary before reaching finalization is shown in \Cref{icc_finalization}.

We point out that if a block $b$ at height $k$ becomes finalized, no other block at height $k$ can be notarized. Since only notarized blocks are extended (see beginning of \Cref{sec:icc}), this implies that all blocks at height greater than $k$ will have $b$ as an ancestor and thus, all \textit{finalized} blocks are part of the same blockchain. This sketches the safety proof of the \ICCal\ algorithm.

\begin{remark}[Finalization Latency]
In rounds in which the network is synchronous, finalization of a block can be reached in three times the message delivery time  $\delta < \Delta$. In the first round the block proposal, in the second the notarization votes, and in the third the finalization votes are broadcast (see \Cref{fig:communication_rounds}).
\end{remark}

\begin{remark}[Timeouts]
In the simplest implementation of the ICC protocol, we can assume that the communication delay bound $\Delta$ is an explicit parameter. In practice, instead, the protocol is modified to adaptively adjust to an unknown communication delay bound.
\end{remark}

\section{Problem Statement}
The purpose of SMR protocols is to totally order blocks containing an arbitrary payload so that all replicas output the payload in the same order. 

The properties we want the \FICC\ algorithm to satisfy, are first and foremost the same that \ICC\ provides~\cite{camenisch2022}:
\begin{itemize}
    \item \label{itm:deadlock-freeness} \textbf{Deadlock Freeness}: Each round eventually terminates and increases the block-tree height by 1. %
    \item \label{itm:safety} \textbf{Safety}: 
    If some honest replica finalizes block $b$ in round $k$, and another honest replica finalizes block $b'$ in round $k$, then $b = b'$.
    \item \label{itm:liveness} \textbf{Liveness}: If the network is momentarily synchronous and the leader is honest, then the block proposed by the leader is added to the block-tree and finalized.
\end{itemize}

Additionally, \FICC\ satisfies a stronger property:
\begin{itemize}
    \setcounter{enumi}{3}
    \item \label{itm:fast-termination} \textbf{Fast Termination}: If the network is momentarily synchronous, the leader is honest, and $n-p$ replicas behave momentarily honestly, then the block proposed by the leader is added to the block-tree and finalized in a single round trip time. 
\end{itemize}

\begin{remark}
Deadlock freeness is also a \textit{liveness} property. It provides chain growth even in periods of asynchrony, while Liveness guarantees that the entire progress is finalized in periods of synchrony. This distinction underlines the two modes of operation~\cite{camenisch2022}.
\end{remark}

\section{Intuition}

The goal of the \FICC\ fast path is to finalize a block as soon as possible while guaranteeing agreement among replicas. In the \ICCal\ protocol a block is finalized at the fastest three communication steps after being proposed. The fast path added in \FICCal\ reduces the finalization latency down to two communication steps, as shown in \Cref{fig:communication_rounds}.

\begin{definition}
In order to distinguish the two possible finalization scenarios, we call \textbf{Fast Path finalized} (or \textbf{FP-finalized}) the blocks that are explicitly finalized via the \FICCal\ fast path, while we call \textbf{Slow Path finalized} (or \textbf{SP-finalized}) the blocks that are explicitly finalized by receiving at least $n-f$ finalization votes, as in the \ICCal\ protocol.
\end{definition}

The idea behind the fast path is that if ``enough'' replicas send their \textbf{first} notarization votes for the same block $b$, i.e., replicas are in pre-agreement, then $b$ can be immediately finalized without waiting for the finalization votes to be sent and received. In other words, replicas finalize blocks through finalization votes only if there are too many replicas that do not agree on the first block added to the block-tree at a given height, or if the slow path is faster. 

\begin{definition}[Fast Vote]
To determine if enough replicas are in pre-agreement, and try to finalize a block via the fast path, replicas broadcast a \textbf{fast vote} for each round's first block they send a notarization vote for. A block that receives $n-p$ fast votes becomes explicitly finalized via the fast path, and fast votes can be combined into a \textbf{fast finalization}.
\end{definition} 

\textbf{Safety considerations.} As we have sketched in \Cref{sec:icc}, the rules of \ICCal\ ensure that if a block $b$ at height $k$ becomes SP-finalized, no other block at the same height will be notarized. This guarantees that the tree will not contain any block at height $k$ besides block $b$. Therefore, only block $b$ will be extended. In \FICCal\ however, if a block $b$ gets FP-finalized, we cannot guarantee that there will not be another notarized blocks at the same height. 

Hence, we introduce a new concept, similar to notarization. Intuitively, we want to guarantee that whenever a block $b$ at height $k$ gets FP-finalized, then no other block at height $k$ can be extended. To this end, each replica marks blocks as \textit{unlocked} when they are safe to be extended.

As we will show, the conditions of a block being \textit{unlocked} almost never restrict the original \ICCal\ algorithm (see \Cref{remark:cornercases}) and only a few additions are sufficient to guarantee the safety of \FICCal.

\section{The \FICC\ Algorithm}
The \FICCal\ algorithm extends the \ICC\ algorithm and enables fast path finalization in periods of synchrony. The protocol allows explicit finalization to be achieved as soon as $n-p$ fast votes are received.  The fast path runs alongside the \ICC\ protocol and is integrated into existing messages. In case the fast path cannot be used, there is no ``switching cost'' to revert to the slow path. When there is no pre-agreement, blocks will still be notarized and eventually finalized (either explicitly or implicitly) just as in the \ICCal\ protocol.

The following definitions are written from the point of view of a replica $u$, at round $k$:
\begin{definition}
Let $\textsc{blocks}(k)$ be the set of blocks that have been received in round $k$. 
For a block $b \in \textsc{blocks}(k)$, we define $\textsc{supp}(b)$ to be the set of replicas from which $u$ has received a fast vote.
Further, for any set $\mathcal{B} \subseteq \textsc{blocks}(k)$, we define $\mathbf{\textsc{supp}(\mathcal{B})}$ to be the set of \textit{distinct} replicas from which $u$ has received a \textit{fast vote} for a block in $\mathcal{B}$:
\begin{equation*}
    \textsc{supp}(\mathcal{B}) = \bigcup_{b \in \mathcal{B}} \textsc{supp}(b)
\end{equation*}
\end{definition}

\begin{definition}
At any given time, $\textsc{max}(k)$ is a rank-0 block, such that no other rank-0 block has a larger support:
\begin{align*}
    \textsc{max}(k) = \left( \argmax_{b \in \textsc{blocks}(k), b.rank=0} \abs{\textsc{supp}(b)} \right ).pop()
\end{align*}
\end{definition}

\begin{remark}
There can be multiple rank-0 blocks when the leader is Byzantine.
\end{remark}

\begin{definition}
At any given time, $\textsc{nonLeaderBlocks}(k)$ is the set of blocks which $u$ has received with a rank larger than 0:
\begin{equation*}
    \textsc{nonLeaderBlocks}(k) = \{b \in \textsc{blocks}(k) \mid b.rank \neq 0 \}
\end{equation*}

\end{definition}

\begin{definition}
At any given time, $\textsc{nonMaxBlocks}(k)$ is the set of blocks which $u$ has received, excluding a rank-0 block that has the most support: 
\begin{equation*}
    \textsc{nonMaxBlocks}(k) = \{b \in \textsc{blocks}(k) \mid b \neq \textsc{max}(k)\}
\end{equation*}
\end{definition}

Using these definitions, we can pin down exactly what blocks are safe to extend for a replica $u$, in the presence of the fast path. We call such a block \textit{unlocked}.

\begin{definition}[Unlocked Blocks]
\label{def:unlocked}
Finalized blocks are \textbf{unlocked} by definition. 
Additionally, at any point during a round $k$, for a valid block $b \in \textsc{blocks}(k)$:
\begin{enumerate}
    \item if $\abs{\textsc{supp}(b) \cup \textsc{supp}(\textsc{nonLeaderBlocks}(k))} > f + p$, then $b$ is unlocked.
    \item if $\abs{\textsc{supp}(\textsc{nonMaxBlocks}(k))} > f + p$, all current and future blocks of round $k$ are unlocked.
\end{enumerate}
\end{definition}

\noindent
An example of the above definitions is presented in \Cref{fig:example_fastable}.

\begin{figure}[h]
\centering
\includegraphics[width=0.42\textwidth]{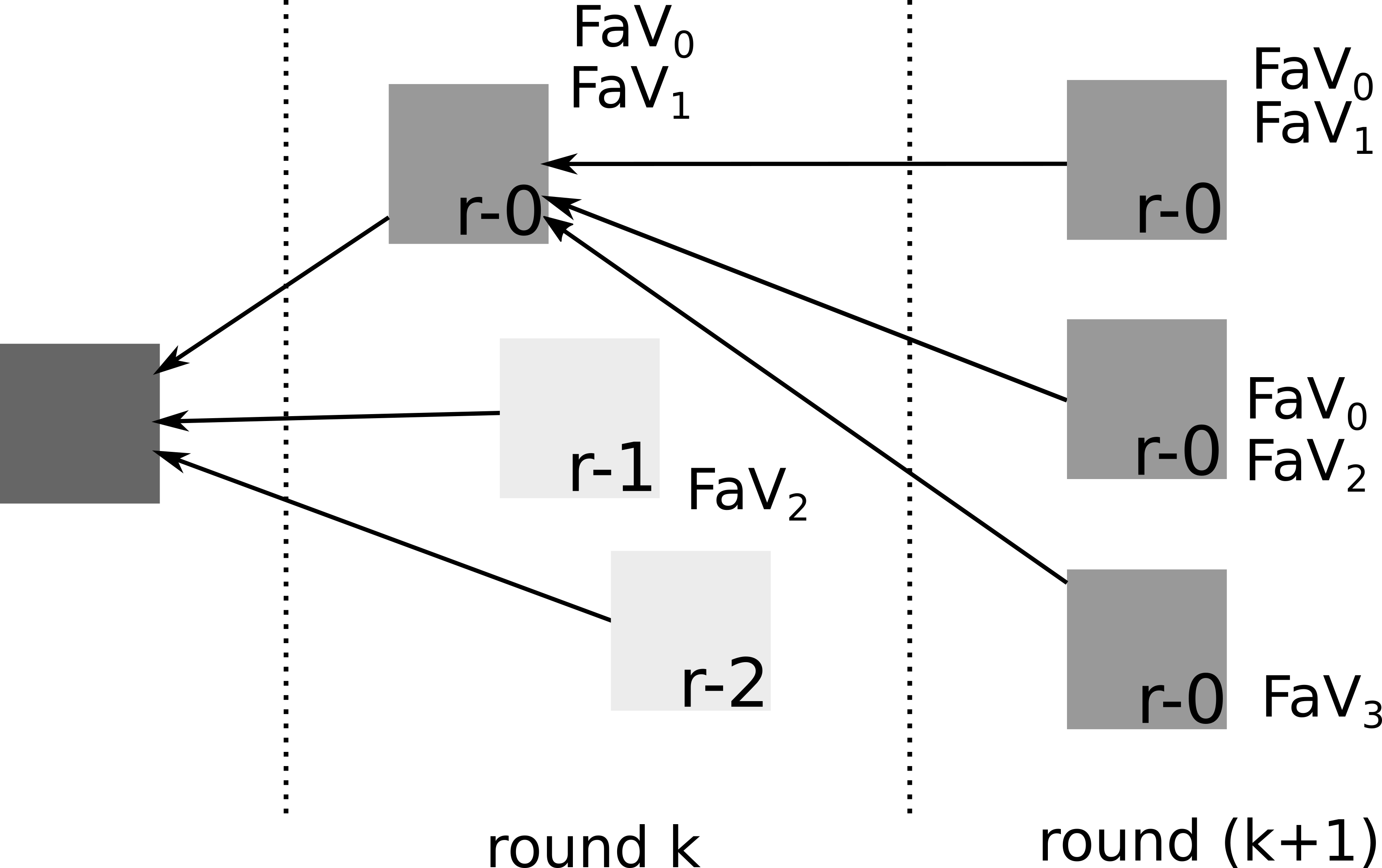}
\caption{Fast votes (FaV) received for each block are shown in the block-tree. Black blocks are \textit{finalized}, dark gray blocks are \textit{unlocked}, while light gray blocks are not. Supposing $n=4$, $f=1$, and $p=1$, for round $k$, Condition 1 is met, and the r-0 block is \textit{unlocked}. Instead, for round $(k+1)$, Condition 2 is met, and all blocks in this round are \textit{unlocked}.} \label{fig:example_fastable}
\end{figure}

\begin{definition}[Unlock Proof]
An \textbf{unlock proof}  is the collection of valid fast votes, that prove that $b$ is unlocked according to \Cref{def:unlocked}. 
\end{definition}
Unlock proofs can be implemented naively by aggregating the fast votes for each block using BLS multi-signatures~\cite{boneh2018compact}. In the worst case, condition 2) of \cref{def:unlocked} might only be met after receiving $2f + 2p + 1$ fast votes, which might attest $f + p + 2$ unique blocks, leaving little improvements from aggregation. %
In practice, however, we expect only a few different blocks to receive fast votes. The description of a more targeted mechanism to achieve small unlock proofs in the worst case is left to future work. %

The \FICCal\ algorithm  is defined by the following changes to the slow path algorithm (\ICCal) presented in \Cref{sec:icc}. Line numbers refer to \Cref{algorithm:ficcal:part1,algorithm:ficcal:part2} that contain the full \FICCal\ pseudocode.

\begin{enumerate}[font=\sffamily, leftmargin=6em]
    \item[Restriction 1] \label{itm:restriction:1} Block proposals can only extend an unlocked block. Similarly, notarization votes, fast votes, and finalization votes are only sent for blocks that extend an unlocked block. (We change the validity condition on \cref{alg:line:validity}.)
    \item[Restriction 2] \label{itm:rule-5} Replicas move to the next round once an unlocked block has been notarized, and they have sent a fast vote (\cref{alg:line:moveifstatement}). 
    \item[Addition 1] \label{itm:rule-6} When replicas move to the next round, an unlock proof of the notarized block is broadcast (\cref{alg:line:proposal}).
    \item[Addition 2] \label{itm:addition:2} When a block is proposed, an unlock proof of the parent block and a fast vote for the current block is broadcast alongside the proposed block (\cref{alg:line:proposal,alg:line:appendfastvote}).
    \item[Addition 3] \label{itm:rule-1} When the first round $k$ notarization vote is broadcast, a fast vote for the same block is broadcast alongside it (\cref{alg:line:bcfastvote}). 
    \item[Addition 4] \label{itm:rule-2} If a rank-0 block has received $n-p$ fast votes, it is considered FP-finalized (\cref{alg:line:finalizationreq}). Fast votes are aggregated into a fast finalization and broadcast (\cref{alg:line:combinefastvotes,alg:line:bcfinalization}).
\end{enumerate}

\begin{remark}
It is possible to omit sending a corresponding notarization vote when a fast vote is sent. A notarization then consists of two multi-signatures, one for notarization and one for fast votes. For the sake of simplicity, we do not consider this version of \FICCal\ in our description and analysis.
\end{remark}

\begin{algorithm}
\begin{algorithmic}[1]
\Implements
    \State Pseudocode for round $k$ at replica $u$
\EndImplements

\Uses
     \Instance{RandomBeacon}{beacon}
     \Instance{BestEffortBraodcast}{broadcast}
\EndUses

\Parameters
    \State $\Delta_{prop}$              \Comment{Proposal delay}
    \State $\Delta_{notary}$            \Comment{Notarization delay}
    \State $k$                          \Comment{Current round number}
    \State $kMax$                       \Comment{Highest round finalized so far}
    \State $payload$                    \Comment{Set of transactions for this round}  
    \State $r_u$                        \Comment{Rank, permutation derived from beacon}
    \State $b_p$                        \Comment{Notarized and unlocked round ($k-1$) block} 
\EndParameters

\UponPure{\ficcal}{Init}{}
    \State $fastVoteSent \gets \false$;    %
    \State $proposed \gets \false$;    \Comment{Proposed a block}
    \State $t_0 \gets clock()$;        \Comment{Time at the start of the round}
    \State $N \gets \emptyset$;        \Comment{Set of blocks for which a notarization vote was sent}
\EndUponPure

\UponCondition{\neg proposed \et clock() \geq t_0 + \Delta_{prop}(r_u)}
    \State \textbf{create} a new round k block:
    \State $b \gets (k, u, hash(b_p), payload, {signature}_u)$ \label{alg:line:test}
    \State $proposed \gets \true$
    \If{$r_u = 0$}
        \State \textbf{broadcast} $b$, $b_p$'s notarization, $b_p$'s unlock proof, fast vote for $b$ \label{alg:line:appendfastvote}%
        \State $fastVoteSent \gets \true$
    \EndIf
    \State \textbf{broadcast} $b$, $b_p$'s notarization, $b_p$'s unlock proof \label{alg:line:proposal}%
\EndUponCondition

\UponExists{\textit{a valid round $k$ block $b$ of rank $r$}}{b \notin N, %
\et clock \geq t_0 + \Delta_{notary}(r) \et \nexists \textit{ valid round $k$ block $b'$ of rank $r'$, $r' < r$} } \label{alg:line:conditionnotarize}
    \If{$r_u \neq r$}
        \State \textbf{broadcast} $b$, $b_p$'s notarization, $b_p$'s unlock proof
    \EndIf
        \State $N \gets N \cup \{b\}$
        \If{$\neg fastVoteSent$}: \label{alg:line:bcfastvote0}
            \State \textbf{broadcast} fast and notarization vote for $b$; \label{alg:line:bcfastvote}
            \State $fastVoteSent \gets \true$; \label{alg:line:bcfastvote2}
        \Else
            \State \textbf{broadcast} notarization vote for $b$
        \EndIf
\EndUponExists

\caption{\FICCal\ (Part 1)}
\label{algorithm:ficcal:part1}
\algstore{ficcal}
\end{algorithmic}
\end{algorithm}

\begin{algorithm}
\begin{algorithmic}[1]
\algrestore{ficcal}

\UponExists{\textit{a valid round $k$ block $b$} \newline}{\textit{$b$ not notarized} \newline
            \et {\ceil{\frac{n+f+1}{2}}} \textit{ notarization votes received}}
    \State \textbf{combine} notarization votes into a notarization for $b$
\EndUponExists

\UponExists{\textit{a valid round $k$ block $b$} \newline}{\textit{$b$ is notarized} \et \textit{$b$ is unlocked} \newline \et fastVoteSent} \label{alg:line:moveifstatement}
    \State \textbf{combine} fast votes into a unlock proof
    \State \textbf{broadcast} notarization and unlock proof for $b$ \label{alg:line:bcfordeadlockfreeness}
    \If{$N \subseteq \{b\}$} \label{alg:line:bcfinalizationvotecheck}
        \State \textbf{broadcast} finalization vote for $b$ \label{alg:line:bcfinalizationvote}
    \EndIf
    \State $k \gets k + 1$  \Comment{Move to next round}
\EndUponExists

\UponExists{\textit{(fast) finalization for a round $k$ block $b$} \newline \textit{ with $k > kMax$} \newline \textbf{\;or\;} \textit{receive } \ceil{\frac{n+f+1}{2}} \textit{ finalization votes} \newline \textbf{\;or\;} \left(\textit{receive } n-p \textit{ fast finalization votes} \newline \textbf{\;and\;} b.rank = 0 \right)\newline}{valid(b) \et \textit{$b$ is not finalized}} \label{alg:line:finalizationreq}
    \State \textbf{combine} (fast) finalization votes into a (fast) \newline finalization for $b$, if necessary \label{alg:line:combinefastvotes}
    \State \textbf{broadcast} (fast) finalization for $b$\label{alg:line:bcfinalization}
    \State \textbf{output} payloads of the last $k-kMax$ blocks in the chain ending at $b$
    \State  $kMax \gets k$
\EndUponExists

\Procedure{valid}{b} \label{alg:line:validity}
\State \textbf{return} $b$ extends a notarized and unlocked round $(k-1)$ block $b_p$ \textbf{and} $b$ is signed correctly \textbf{and} contains a fast vote from the proposer if $b.rank = 0$
\EndProcedure
\caption{\FICCal\ (Part 2)}
\label{algorithm:ficcal:part2}
\end{algorithmic}
\end{algorithm}

\section{Protocol Analysis}
In this section, we present proof sketches that explain the correctness of \FICC. These are not complete proofs but are intended to provide a conceptual understanding of the protocol's security.

\subsection{Deadlock Freeness}
The \ICCal\ protocol guarantees deadlock freeness as each honest replica receives at least one notarized block in each round. Therefore, each honest replica can start the next round and the block-tree keeps growing.
However, in \FICCal\, a block that was notarized is not guaranteed to be extended as, due to \textsf{Restriction 1}, in order for this to happen, the block must also be unlocked. 

The following lemma shows how, regardless, deadlock freeness is not impacted. Intuitively, \textsf{Restriction 2} makes sure that replicas only move to a higher round once a block can be extended, and \textsf{Addition 1} guarantees that each replica does so eventually.

\begin{lemma}
\label{lemma:df:oneisunlocked}
If each honest replica broadcasts a fast vote in round $k$, then \FICCal\ guarantees that each honest replica eventually observes at least one unlocked block.
\end{lemma}
\begin{proof}
Assume towards contradiction that in round $k$ no block is unlocked for an honest replica $u$. 
Consider the set $\mathcal{S}$, defined to be the support for leader blocks different from $\textsc{Max}(k)$, i.e., 
$$\mathcal{S} = {\textsc{supp}(\textsc{nonMaxBlocks}(k) \setminus \textsc{nonLeaderBlocks}(k))}$$

If $\mathcal{S} = \emptyset$, a contradiction is reached immediately, as
\begin{align*}
& \abs{\textsc{supp}(\textsc{max}(k)) \cup \textsc{supp}(\textsc{nonLeaderBlocks}(k))}\\ 
&= \abs{\textsc{supp}(\textsc{blocks}(k))} \geq n - f > f + p    
\end{align*}
and thus by Item 1 of \Cref{def:unlocked} $\textsc{max}(k)$ would be unlocked. (This corresponds to rounds with an honest leader.)

Instead, if $\mathcal{S} \neq \emptyset$, there are at least two rank-0 blocks. By Addition 2 (\Cref{alg:line:appendfastvote}), they each contain a fast vote from the (Byzantine) leader. Since $u$ also receives fast votes from $n-f \geq 2f + 2p - 1$ honest replicas, at least $2f + 2p + 1$ fast votes will be received. This implies that 
\begin{equation*}
\abs{\textsc{supp}(\textsc{max}(k))} + \abs{\textsc{supp}( \textsc{nonMaxBlocks}(k))} \geq 2f + 2p + 1
\end{equation*}
By the pigeonhole principle, either $\abs{\textsc{supp}(\textsc{max}(k))} > f + p $ or $\abs{\textsc{supp}(\textsc{nonMaxBlocks}(k))} > f + p $. In both cases at least one block is unlocked according to \Cref{def:unlocked}, leading to a contradiction. 
\end{proof}

\begin{theorem}
    \FICCal\ satisfies \textbf{deadlock freeness}.
\end{theorem}

\begin{proof}
We prove that each round eventually terminates, and that the block-tree of each honest replica keeps growing in height by induction on the block-tree height $k$.
Specifically, we show that for each round $k$, at least one notarized and unlocked block will be added to each honest replica's block-tree. 

\textbf{Base case:} All honest replicas agree on the root of the block-tree, the genesis block, which is also notarized and finalized by definition. By \Cref{def:unlocked} the genesis block is thus unlocked. Thus, in round 0, a notarized and unlocked block is added to each replica's block-tree.

\textbf{Induction step:} Assuming a notarized and unlocked block exists in round $k$, we prove that each replica observes a notarized and unlocked block in round $k+1$. By the induction hypothesis, all replicas will broadcast a fast vote and thus execute the procedure that starts on \cref{alg:line:moveifstatement}. Thus, all honest replicas will enter round $k+1$.

Each honest replica will receive at least one round $k+1$ block, together with a notarization and unlock proof for the extended parent block (\textsf{Addition 2}, \cref{alg:line:proposal}). Thus, each honest replica will send a fast vote for one block (\textsf{Addition 3}, \cref{alg:line:bcfastvote}). Thus, by \Cref{lemma:df:oneisunlocked}, at least one block round $k+1$ block will be unlocked.
As no honest replica moves to a higher round without an unlocked block being notarized (\textsf{Restriction 2}, \cref{alg:line:moveifstatement}), and replicas keep notarizing blocks, one unlocked block will be notarized. Finally, honest replicas moving from one round to the next broadcast the notarization and unlock proof (\textsf{Addition 1}, \cref{alg:line:bcfordeadlockfreeness}), guaranteeing that all honest replicas can add a notarized and unlocked block to their block-tree.
\end{proof}

\begin{remark}
\label{remark:cornercases}
    \FICC\ is at least as fast as \ICCal. Any scenario in which a notarized block for round k exists before it is unlocked can only occur due to message reordering. Hence, \textsf{Restriction 1} and \textsf{Restriction 2} cannot cause latency concessions if the communication channel precludes message reordering (which is the case in practice when TCP/QUIC is used). 
\end{remark}

\subsection{Safety}

\ICCal\ derives its safety from the fact that in rounds with an explicetly finalized block, no other blocks can be notarized. We start our way towards sketching the safety proof of \FICCal\ by showing that an equivalent property holds for \FICCal.

\begin{lemma}
\label{lemma:ic_fin_no_notar}
If a round $k$ block $b$ is \textit{SP-finalized}, then \FICCal\
guarantees that no round $k$ block $b'$, $b'\neq b$ is {notarized} for any replica.
\end{lemma}

\begin{proof}
Let block $b$ be SP-finalized by some replica, and let block $b'$, $b' \neq b$ be notarized. A replica must have received a quorum of at least $\ceil{\frac{n+f+1}{2}}$ finalization votes for $b$, $\ceil{\frac{n-f+1}{2}}$ of which are from honest replicas. 
Moreover, some replica must have received a quorum of $\ceil{\frac{n+f+1}{2}}$ notarization votes for $b'$,  $\ceil{\frac{n-f+1}{2}}$ of which are from honest replicas.
The two quorums must be disjoint, since each honest replica sends a finalization vote for $b$ only if it did not send any notarization vote for other blocks (in particular $b'$) at height $k$ (\cref{alg:line:bcfinalizationvotecheck}). 
Thus, there must be at least $\ceil{\frac{n-f+1}{2}} + \ceil{\frac{n-f+1}{2}} \geq n-f+1$ honest replicas. By definition, there are only $n-f$ honest replicas, hence one honest replica must occur in both quorums, a contradiction.
\end{proof}

Therefore, any other block at height $k$ will be neither extended nor finalized (as both cases require it to be at least notarized). Similarly, let $b$ be a block at height $k$ which is FP-finalized, we prove that $b$ is the only unlocked block at height $k$. 

\begin{lemma} 
\label{lemma:if_fp_fin_no_other_good}
If a round $k$ block $b$ is \textit{FP-finalized}, then \FICCal\ guarantees that no round $k$ block $b'$, $b'\neq b$ is unlocked for any replica.
\end{lemma}
\begin{proof}

Assume towards contradiction that there exists an FP-finalized block $b$ and an unlocked block $b'$, $b' \neq b$ at height $k$. 
For $b'$ to be unlocked, either Condition 1 or Condition 2 of \Cref{def:unlocked} must be satisfied.

If Condition 1 is satisfied, then $$\abs{\textsc{supp}(b') \cup \textsc{supp}(\textsc{nonLeaderBlocks}(k))} > f + p$$ This implies that blocks different from $b$ have received more than $f+p$ fast votes, more than $p$ of which coming from honest replicas. By \cref{alg:line:bcfastvote0,alg:line:bcfastvote,alg:line:bcfastvote2}, each honest replica broadcasts at most one fast vote per round. Thus $b$ received less than $n-p$ fast votes, a contradiction, as $n-p$ fast votes are required for FP-finalization (\cref{alg:line:finalizationreq}).

If Condition 2 is satisfied, then $$\abs{\textsc{supp}(\textsc{nonMaxBlocks}(k))} > f + p$$ It follows that $\abs{\textsc{supp}(\textsc{max}(k))} < n - p$. This implies that no block can be FP-finalized, a contradiction. 
\end{proof}

\begin{theorem}
\label{thm:ficc_safety}
    \FICCal\ satisfies \textbf{safety}.
\end{theorem}
\begin{proof}
Towards contradiction, assume an honest replica finalizes round $k$ block $b$, while another honest replica finalizes round $k$ block $b'$, $b' \neq b$. Without loss of generality, we assume $b$ was explicitly finalized. 

Consider the case where $b$ and $b'$ were both explicitly finalized. SP-and FP-finalization imply the reception of $\ceil{\frac{n+f+1}{2}}$ unique notarization and fast votes for a block, respectively. This implies the existence of two respective Byzantine quorums, a contradiction~\cite{textbook}.
Thus, without loss of generality, assume $b'$ was implicitly finalized. If $b$ was SP-finalized, by \Cref{lemma:ic_fin_no_notar} $b'$ cannot have been notarized, and thus no honest replica would have extended $b'$.
Instead, if $b$ was FP-finalized, by \Cref{lemma:if_fp_fin_no_other_good} then $b'$ was not unlocked. Thus, honest replicas would have sent fast, notarization, or finalization votes for a block that was not unlocked, a contradiction (see \cref{alg:line:conditionnotarize} and the validity definition on \cref{alg:line:validity}). 
\end{proof}

\subsection{Liveness}
\label{sec:liveness}
We define liveness such that whenever the network is synchronous for a “long enough” interval, if the leader is honest, only the leader’s block will be added to the block-tree for the corresponding round and this block will be finalized by all honest replicas.

\begin{theorem}   
    \FICCal\ satisfies \textbf{liveness}.
\end{theorem}

\begin{proof}

Let $T$ be the time when the first honest replica $u$ enters round $k$. Suppose that the leader $l$ of round $k$ is honest and proposes block $b$. Moreover, suppose that the network is $\delta$-synchronous between time $T$ and $T+2\delta$.
Furthermore, assume $\delta \leq \Delta$, which implies that $\Delta_{notary}(1) \geq 2 \delta $, with the functions $\Delta_{notary}$ and $\Delta_{prop}$ as defined in \Cref{sec:icc}. 
We want to prove that under these assumptions, eventually, each honest replica will finalize block $b$.

Before the honest replica $u$ enters round $k$ at time $T$, the replica $u$ notarized an unlocked block in rounds $1, \dots, k-1$ and broadcast the notarization and unlock proof (\textsf{Addition 1}, \cref{alg:line:bcfordeadlockfreeness}). By the synchrony assumption, the other replicas, $l$ included, receive the notarization and unlock proof for rounds $1, \dots, k-1$ by time $T+\delta$, and enter round $k$. 

Note that $\Delta_{prop}(0) = 0$. Thus, replica $l$ broadcasts $b$ by time $T+\delta$, and other replicas will receive it by $T+2\delta$. By assumption, no other honest replica enters round $k$ before time $T$ and as $\Delta_{notary}(1) \geq 2 \delta $, every honest replica will broadcast a fast vote and a notarization vote for $b$ (\textsf{Addition 3}, \cref{alg:line:bcfastvote}). Therefore, all honest replicas eventually receive at least $n-f$ fast votes for $b$. When $p=f$, these fast votes are enough to FP-finalize $b$ (\textsf{Addition 4}, \cref{alg:line:finalizationreq}). When $p<f$, $\ceil{\frac{n+f+1}{2}} \leq n-f$ notarization votes are sufficient to notarize $b$ instead. As argued previously, honest replicas will not have notarized another block and thus will broadcast a finalization vote (\cref{alg:line:bcfinalizationvote}). Eventually, after an additional communication round, honest replicas will receive $\ceil{\frac{n+f+1}{2}} \leq n-f$ finalization votes, which leads to SP-finalization (\cref{alg:line:finalizationreq}).
\end{proof}

\subsection{Fast Termination}

\begin{theorem}
    \FICCal\ satisfies \textbf{fast termination}.
\end{theorem}
\begin{proof}
With the additional assumption that $n-p$ replicas behave honestly during a phase of synchrony, we can guarantee $n-p$ fast votes for the rank-0 block reaching every replica. Thus, after a single round-trip time, each replica will be able to FP-finalize the rank-0 block.
\end{proof}

\section{Evaluation}
\label{sec:evaluation}

\subsection{Implementation}
\label{sec:evaluation:implementation}

We implement \FICCal\ in Golang, utilizing the Bamboo BFT framework, which is designed for prototyping and evaluating chained BFT algorithms~\cite{gai2021bamboo}. This framework already supports the two well-known protocols, Streamlet~\cite{chan2020streamlet} and HotStuff~\cite{yin2019hotstuff}. 

To enhance the performance and reliability of the framework, we introduced a few modifications. For instance, by using non-blocking message-passing channels internally, and by forwarding blocks that extend the tip of the chain, we drastically improve the performance of all algorithms implemented with Bamboo. The impact of the latter change is especially surprising to us, as originally only Streamlet pursued this strategy in the provided implementation, leading to an unexpected advantage over all other protocols, both in terms of throughput and latency. Henceforth, we have been careful to treat all protocols equally.

To further increase predictability and transparency of the protocol operations, we have replaced the random beacon leader election with a round-robin rotation. 
Our code and benchmarking scripts are made openly available~\cite{yannvon2024gitrepo}.

\subsection{Methodology}
\label{sec:evaluation:methodology}

\begin{figure}[]
\centering
\includegraphics[width=0.45\textwidth]{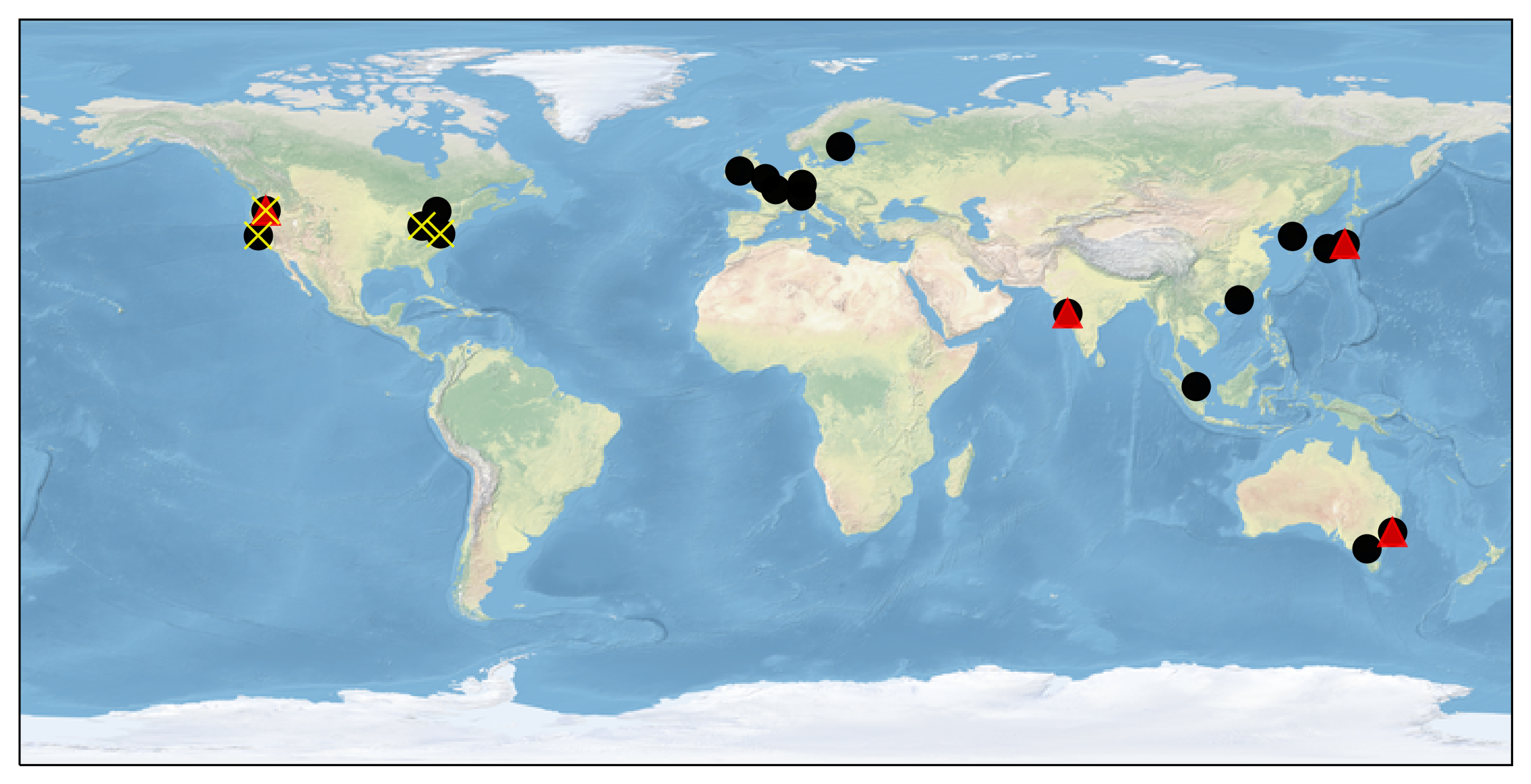}
\caption{Locations of the AWS datacenters used in our testbed. Red triangles represent the 4 datacenters employed in \Cref{exp:equi}, yellow crosses indicate the 4 US datacenters used in \Cref{exp:byz}, and black dots denote the 19 worldwide datacenters used in \Cref{exp:wan}.}
\label{fig:map}
\end{figure}

Application subnets of the Internet Computer feature 13 replicas~\cite{icp_subnets}. To test the potential of \FICC\ as a drop-in replacement for ICC, and since $n=19$ is optimal for both $f=6$, $p=1$ and $f=4$, $p=4$ experimental scenarios, we perform experiments with up to 19 replicas around the globe (see \Cref{fig:map}). Replicas run on a testbed comprised of AWS \texttt{t3.large} EC2 instances (with 2vCPUs, 8 GB of memory, running Ubuntu 22.04).

We wish to demonstrate that (i) requiring only two rounds of communication instead of three has a positive impact and does not increase variance, while (ii) not suffering any more throughput or latency degradation than ICC under crash-faults. Finally, we want to show (iii) that \FICCal\ is performing consistently as well as or better than \ICCal\ and other state-of-the-art protocols in a global network topology.

As the goal of our evaluation is to measure the precise impact of our changes, we define latency as the average proposal finalization time, measured at the respective proposer using their system clocks. Throughput is calculated as the average number of committed bytes per second at any (non-faulty) replica. 
We measure the protocols' behavior under varying load by controlling the size of the payload, where the payload is a bit vector randomly generated by the leader.

As required by the protocol, we set the proposal ($\Delta_{prop}$) and notarization ($\Delta_{notary}$) delays for \FICCal\ and \ICCal\ to be larger than the message delay experienced without network disruptions, such that our experimental results correspond to regular network conditions and only one block is proposed per round if there are no faults. We proceed in the same way to set appropriate timeouts for HotStuff and Streamlet. Each experiment is run for 120 seconds, which is shown to be sufficient, as the measurements show remarkable regularity (e.g., see \Cref{fig:n4:variance}).

\begin{figure*}
\begin{subfigure}[h]{0.33\textwidth}
\includegraphics[width=\textwidth]{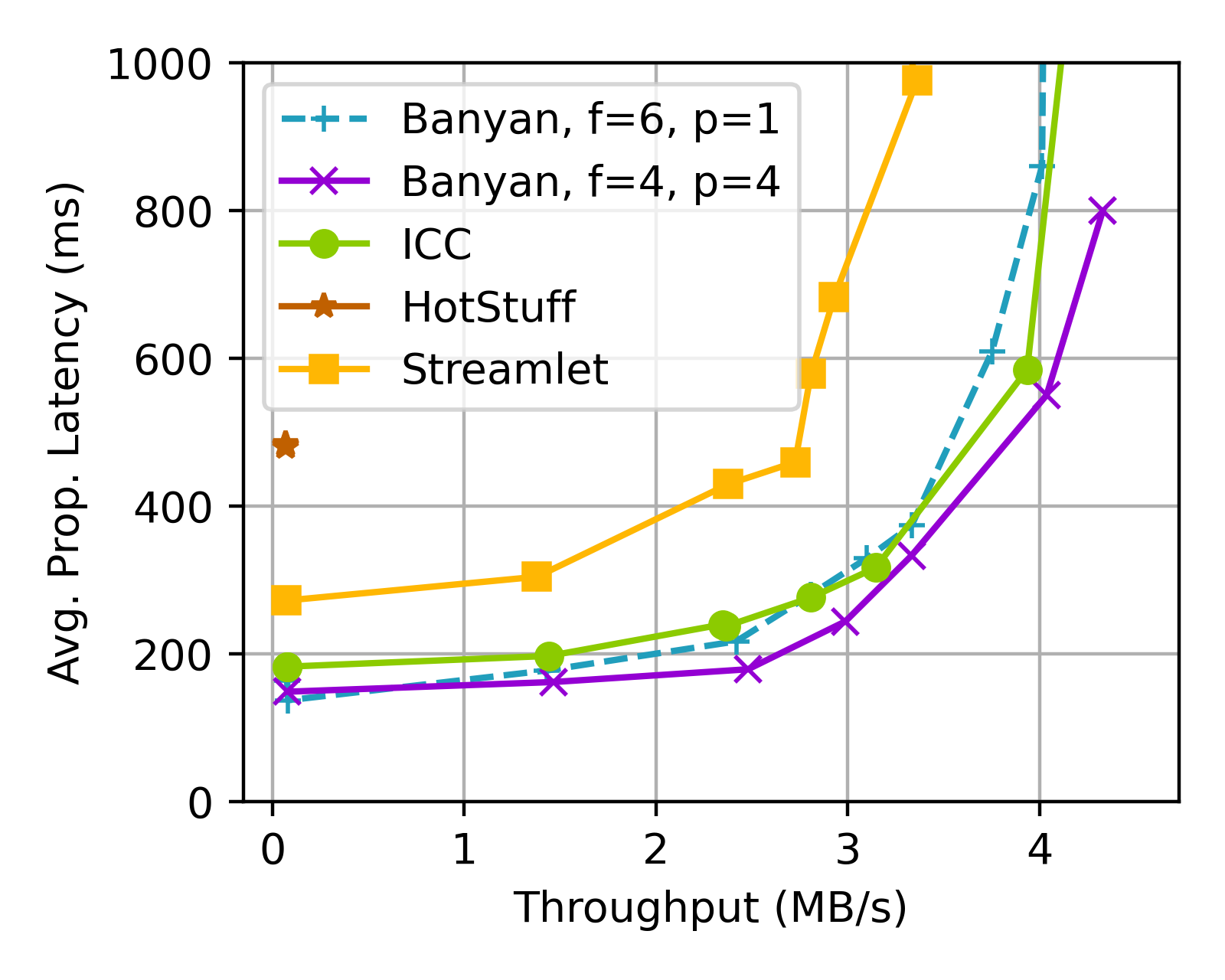}
\caption{Throughput vs. proposal latency for n=19 replicas spread across 4 global datacenters.}
\label{fig:n4_19:prop}
\end{subfigure}
\begin{subfigure}[h]{0.33\textwidth}
\includegraphics[width=\textwidth]{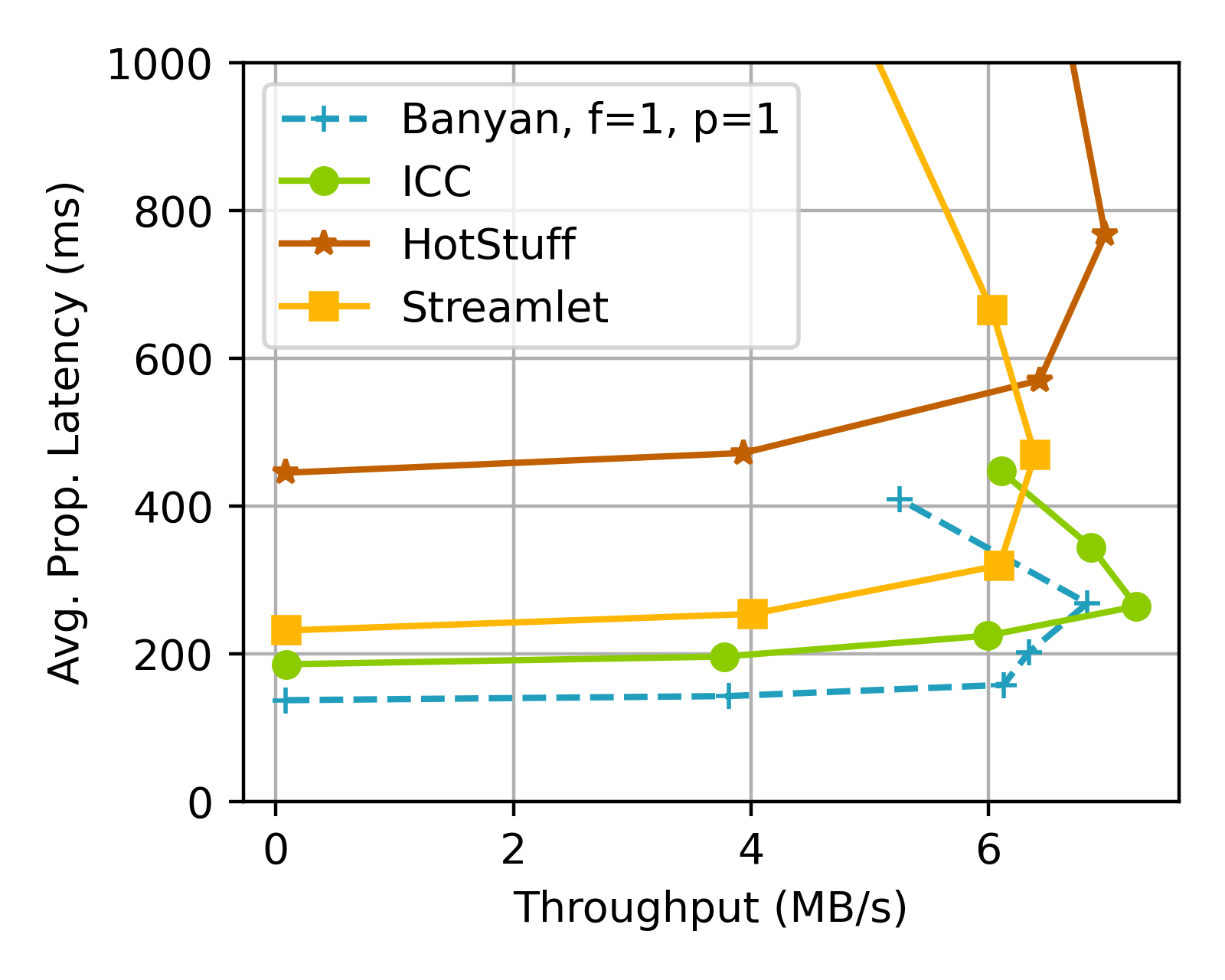}
\caption{Throughput vs. proposal latency for n=4 replicas spread across 4 global datacenters.}
\label{fig:n4:prop}
\end{subfigure}
\begin{subfigure}[h]{0.33\textwidth}
\includegraphics[width=\textwidth]{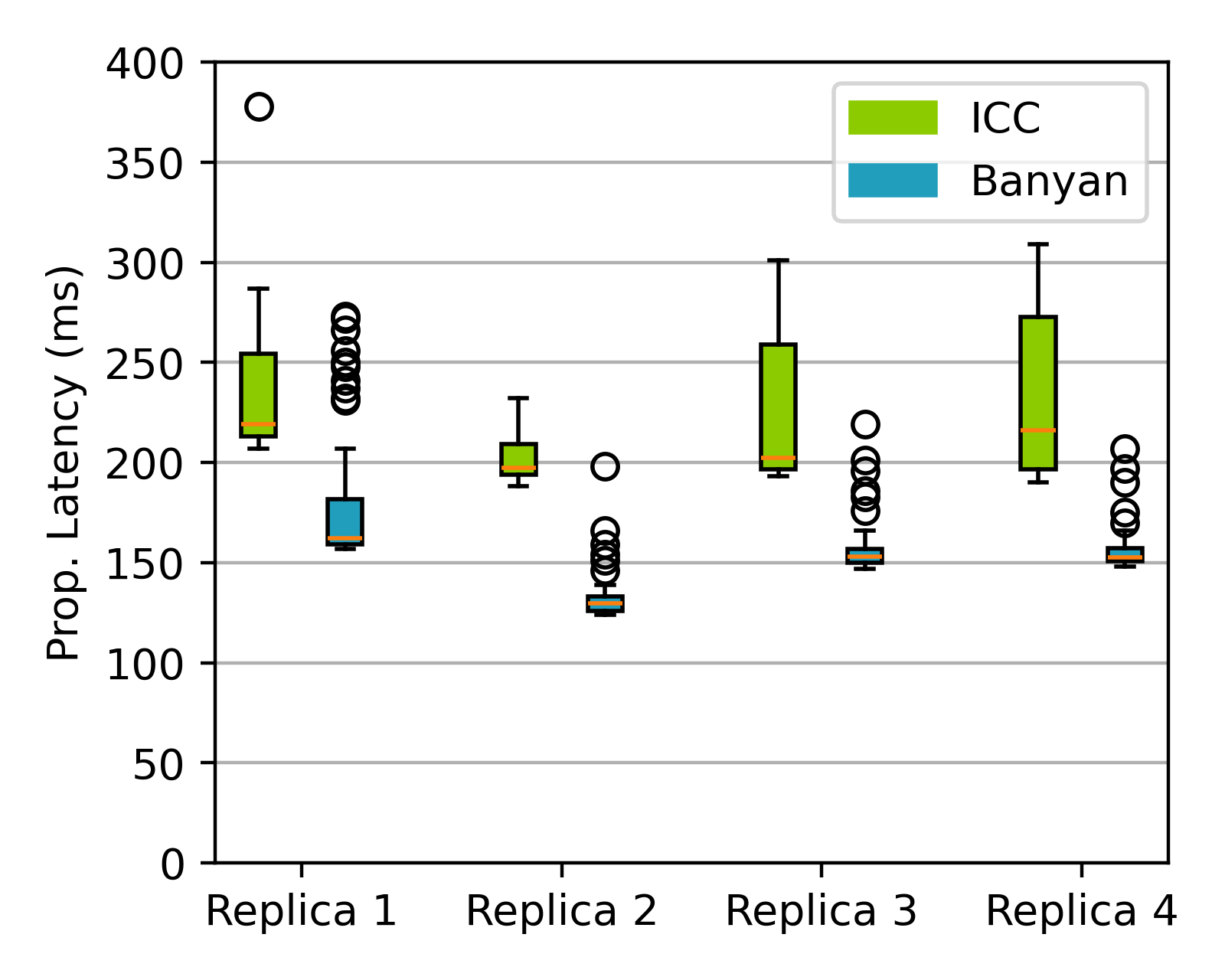}
\caption{Variance of \FICC\ and ICC proposal latencies compared with 1 MB payload and n=4.}
\label{fig:n4:variance}
\end{subfigure}
\begin{subfigure}[h]{0.33\textwidth}
\includegraphics[width=\textwidth]{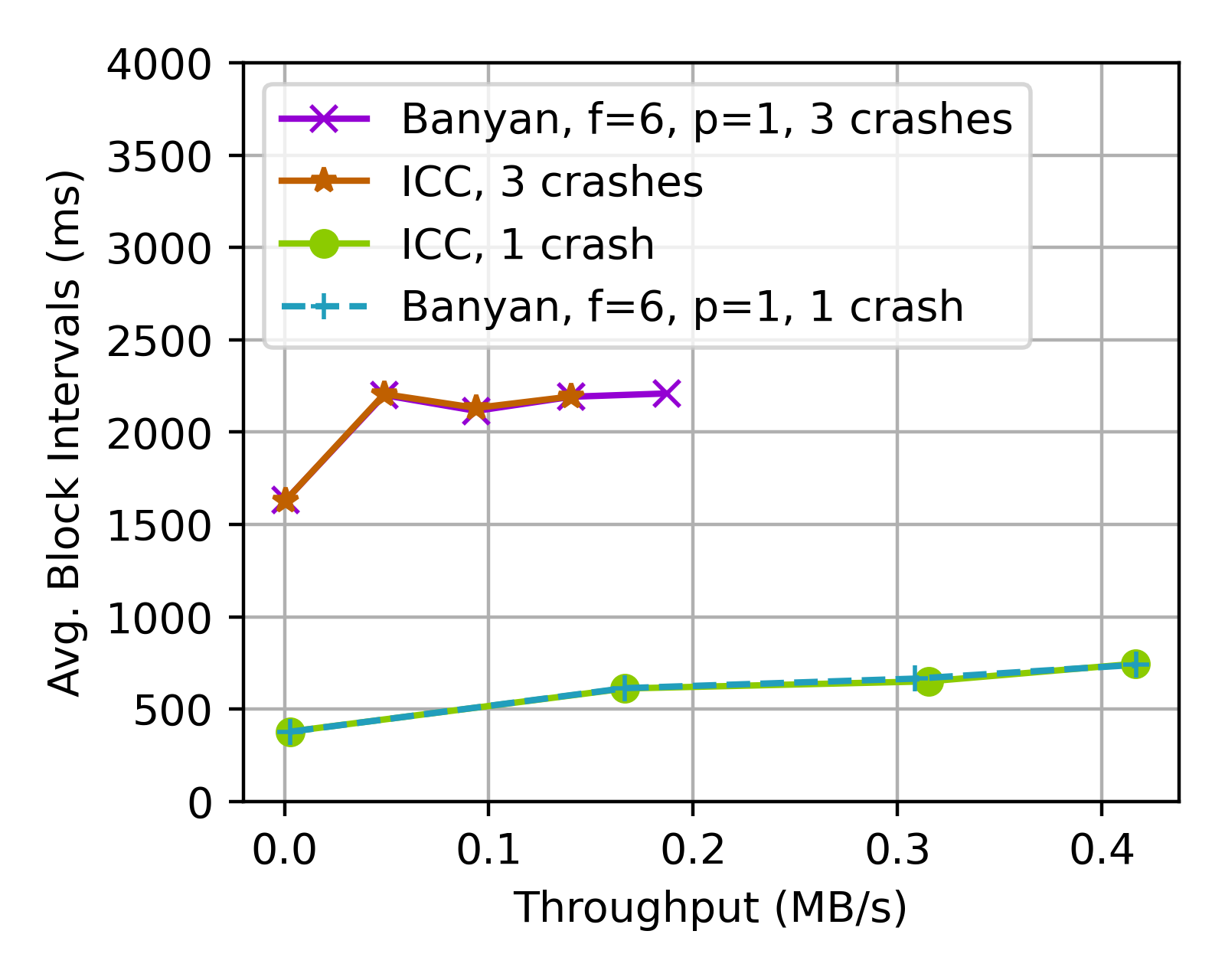}
\caption{Effect of crash-faults on throughput and block intervals for n=19 replicas, spread across 4 US datacenters.}
\label{fig:wan:byz}
\end{subfigure}
\begin{subfigure}[h]{0.33\textwidth}
\includegraphics[width=\textwidth]{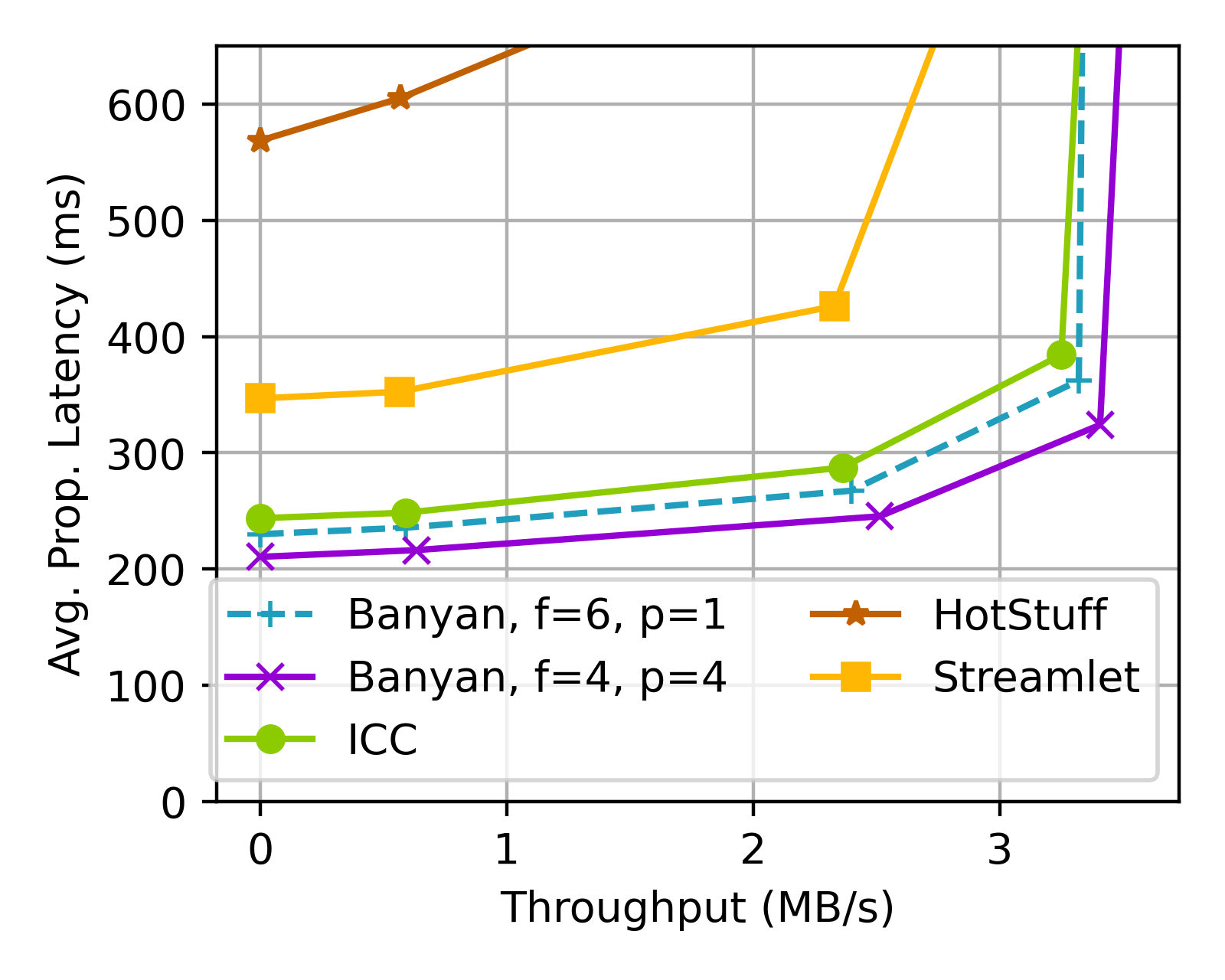}
\caption{Throughput vs. proposal latency for n=19 replicas spread across a global network.}
\label{fig:wan19:prop}
\end{subfigure}
\caption{Evaluation results. }
\label{fig:outlier}
\end{figure*}

\subsection{Performance Evaluation}
\label{exp:equi}

We distribute 19 replicas across 4 datacenters, as shown by the red triangles in \Cref{fig:map}. Each datacenter hosts 5 replicas, except for one that hosts only 4. 
\Cref{fig:n4_19:prop} shows the throughput with varying block sizes. Primarily, we note that \FICC\ with $p=1$ performs consistently better or as well as ICC. For blocks of size 400KB, ICC averages a proposal finalization time of 239ms, while \FICC\ $p=1$ averages a finalization time of 216ms. This improvement of about $10\%$ is expected to be less than the theoretical maximum of $33\%$, because Banyan must hear from all the data centers in the fast path. 
In the slow path, replicas can collect two consecutive quorums without needing to communicate with the furthest datacenter. This observation, namely that more communication rounds with smaller quorum sizes are sometimes faster in practice than fewer rounds with larger quorum sizes, has been made before in similar contexts ~\cite{junqueira2007classic,sousa2015separating}. Moreover, by setting $p=f=4$, we observe that \FICC\ now has an average proposal finalization time of 179ms at the same block size, an improvement of 25.1\%, that is much closer to the theoretical maximum of 33\%. We suggest that this is due to the situation where the 4 co-located replicas are the furthest, in which case the fast path is employed, and thus lowers the average latency.

Next, we measure the potential of the \FICC\ fast path with a low number of replicas, i.e., we use only a single replica at each datacenter ($n=4$). Crucially, in this case the fast path fires with the same conditions as regular notarization, i.e., after receiving just 3 replies. \Cref{fig:n4:prop} shows the throughput with block size increments of 500KB. At block sizes of 1MB, ICC averages a proposal finalization time of 224ms, which \FICC\ manages to improve by 29.9\% to 157ms. 
We emphasize that this large performance increase does not come at the cost of higher variance in latency, as is shown in \Cref{fig:n4:variance}.

\subsection{Effect of Crash-faults}
\label{exp:byz}
In this experiment, we show the effect that crashes have on \FICC\ and \ICC. We distribute 19 replicas across four US datacenters, as shown by the yellow crosses in \Cref{fig:map}. Results are presented in \Cref{fig:wan:byz}.
As has been noted in literature, rotating leader protocols are sensitive to crashes (or malicious behavior), as at least one full timeout duration must expire before progress can be made in the case of a faulty leader. The timeout parameter will crucially impact the results of such scenarios. In the presented experiment, the timeout was set to 3 seconds.
As seen in \Cref{fig:wan:byz}, there are no penalties in trying to take the fast path. When there are failures, the performance of Banyan is exactly the one of ICC.

\subsection{Global Network}
\label{exp:wan}
In this experiment, we simulate a worldwide deployment. We include almost all AWS datacenter locations available to us and distribute 19 replicas across the globe, as shown by the black dots in \Cref{fig:map}.
In this setting and with 1MB payloads, the average proposal finalization time measured for the \ICCal\ implementation was 384ms. 
By running \FICCal\ with $f=6, p=1$, the average proposal finalization time is reduced by 5.8\% to 362ms for ``free''. 
Further, for \FICCal\ with $f=4, p=4$, the number drops by 16\% to 324ms (\Cref{fig:wan19:prop}).

\section{Conclusion}

The proposed \FICC\ protocol improves upon current BFT protocols by providing a simultaneous dual mode, without incurring extra cost, thus closing the gap between state-of-the-art SMR protocols and classical fast consensus literature.

We show that it is possible to achieve optimally fast proposal finalization time for a rotating leader SMR protocol. We experimentally demonstrate the effectiveness of the fast path by showing consistent latency improvements compared to ICC, HotStuff and Streamlet. Ultimately, the effectiveness of the fast path depends largely on the network topology, as well as the parameter $p \in [0,f]$. Still, as we can set $p=1$ at essentially no cost, we hypothesize that \FICC\ can have large positive impact in practice, especially in networks with fewer replicas. In the context of permissioned blockchains, and future decentralized Layer-2 sequencers, this setting is undoubtedly compelling, since it is possible to detect adversarial behavior and remove misbehaving replicas~\cite{buterin2017casper, camenisch2022}. %

\section{Acknowledgements}
We thank the anonymous Middleware reviewers for their extensive and valuable reviews, that among other things lead to the improvement of the resilience of \FICC. We further extend our gratitude to Prof. Christian Cachin for his insightful feedback.

\bibliographystyle{acm-bib-format}
\bibliography{bibliography}

\appendix

\end{document}